\documentclass{acm_proc_article-sp}

\usepackage{enumerate}
\usepackage{booktabs}
\usepackage{array}
\usepackage{upgreek}
\usepackage{latexsym,amssymb,amsmath}
\usepackage{graphicx, color}

\newenvironment{enumpack}{
\vspace{-0.1in}
\begin{enumerate}[(i)]
	\setlength{\partopsep}{0pt}
  \setlength{\itemsep}{1pt}
  \setlength{\parskip}{0pt}
  \setlength{\parsep}{0pt}
}{\end{enumerate}}

\newcolumntype{M}[1]{>{\centering\arraybackslash}m{\dimexpr#1\linewidth}}

\newdef{definition}{Definition}
\newtheorem{theorem}[definition]{Theorem}
\newtheorem{lemma}[definition]{Lemma}
\newtheorem{proposition}[definition]{Proposition}

\renewcommand{\P}{\mathcal{P}}

\newcommand{\F}{\mathcal{F}}
\renewcommand{\O}{\mathcal{O}}
\renewcommand{\emph}{\textsl}
\newcommand*{\euler}{\mathrm{e}}

\DeclareMathOperator{\str}{str}
\DeclareMathOperator{\sstr}{sstr}
\DeclareMathOperator{\vc}{vc}
\DeclareMathOperator{\dvc}{dvc}

\DeclareMathOperator{\dom}{dom}

\DeclareMathOperator{\ap}{ap}

\newcommand{\ignore}[1]{}

\newcommand{\lk}[1]{\comment{\textcolor{blue}{[LK]: #1}}}

\newcommand{\comment}[1]{\textcolor{blue}{\textbf{#1}}}

\newcommand{\powset}{\mathcal{P}}

\newcommand{\setdef}[2]{\{#1\ :\ #2\}}
\newcommand{\s}{\mathbb{S}}

\newcommand{\system}[2]{\left\langle #1,\ \{0,1\}^{#2}\right\rangle }
\newcommand{\systemexp}[2]{\left\langle #1,\ {#2}\right\rangle }
\newcommand{\sys}[2]{\left\langle #1,\ #2\right\rangle }

\newcommand{\mymerge}{\star}
\newcommand{\defeq}{\overset{\underset{\mathrm{\tiny def}}{}}{=}}

\begin{document}


\title{Shattering, Graph Orientations, and Connectivity}
\subtitle{[Extended Abstract]
}

\numberofauthors{2}
\author{
\alignauthor
L\'{a}szl\'{o} Kozma\\
       \affaddr{Saarland University}\\
       \affaddr{Saarbr\"{u}cken, Germany}\\
       \email{kozma@cs.uni-saarland.de}
\alignauthor
Shay Moran\\
       \affaddr{Max Planck Institute for Informatics}\\
       \affaddr{Saarbr\"{u}cken, Germany}\\
       \email{smoran@mpi-inf.mpg.de}
}

\maketitle

\begin{abstract}

We present a connection between two seemingly disparate fields: VC-theory and graph theory. This connection yields natural correspondences between fundamental concepts in VC-theory, such as shattering and VC-dimension, and well-studied concepts of graph theory related to connectivity, combinatorial optimization, forbidden subgraphs, and others.

In one direction, we use this connection to derive results in graph theory. Our main tool is a generalization of the Sauer-Shelah Lemma~\cite{Pajor,BR95,Dress1}. Using this tool we obtain a series of inequalities and equalities related to properties of orientations of a graph. Some of these results appear to be new, for others we give new and simple proofs.

In the other direction, we present new illustrative examples of shattering-extremal systems - a class of set-systems in VC-theory whose understanding is considered by some authors to be incomplete~\cite{BR95,Greco98,S-ext}. These examples are derived from properties of orientations related to distances and flows in networks.

\end{abstract}

\section{Introduction} \label{sec1}





Orientations of graphs have been widely researched, going back to the celebrated strong-orientation theorem of Robbins (1939)~\cite{robbins}, and its generalization by Nash-Williams (1960) \cite{nash-williams}. Several results analogous to Robbins' theorem have been obtained for other properties of digraphs~\cite{gallai68, frank80, ghouila, chartrand, chvatal}. 
In another direction of research concerning orientations, Frank \cite{frank} has shown that every $k$-strong orientation of a graph can be obtained from any other $k$-strong orientation of the same graph, through a sequence of reversals of directed paths and circuits, such that the $k$-strong connectivity is maintained throughout the sequence.
There are many open problems in this area, both of structural and of algorithmic nature. For a comprehensive account of known results and current research questions, especially relating to connectivity, see~\cite{digraphs}.
A different line of work concerns counting orientations with forbidden subgraphs. These problems have close ties with the theory of random graphs. An example of work of this type is the paper of Alon and Yuster~\cite{alon06}, concerned with the number of orientations that do not contain a copy of a fixed tournament.

In this paper we prove statements concerning properties of orientations and subgraphs of a given graph by introducing a connection between VC-theory and graph theory. The following are a few examples of the statements that we prove: 
\begin{enumerate}
\item{Let $G$ be a graph. Then:

\begin{tabular}{M{.25}M{.01}M{.22}M{.01}M{.30}}
the~number of~connected subgraphs {of~$G$} & $ \geq $ & the~number of~strong orientations of~$G$ & $ \geq $ & the~number of~2-edge-connected subgraphs of $G$.
\end{tabular}}
\item{ Let $G$ be a graph, let $\vec H$ be a digraph, and let $H$ be the graph that underlies $\vec H$. Then:

\begin{tabular}{M{.4}M{.01}M{.4}}
the number of orientations of $G$ that~do~not~contain~a copy~of~$\vec H$ & $\leq$ & the number of subgraphs of $G$ that~do~not~contain~a copy~of~$H$.
\end{tabular}}
\item{Let $G$ be a graph and let $s,t\in V(G)$. Then:

\begin{tabular}{M{.45}M{.01}M{.45}}
the number of orientations in which there are at least $k$ edge-disjoint paths from~$s$~to~$t$ & = & the number of subgraphs in which there are at least $k$ edge-disjoint paths from~$s$~to~$t$.
\end{tabular}}

(A similar statement holds for vertex-disjoint paths.)
\item{Let $O'$ and $O''$ be two orientations of a flow network $N$ such that there exist $s $ - $ t$ flows of size $f$ in both $O'$ and $O''$, and let $d$ be the \emph{Hamming distance} between $O'$ and $O''$. Then there exists a sequence of $d$ \emph{edge flips} that transforms $O'$ to $O''$, such that all the intermediate orientations preserve the property of having a flow of size $f$.}
\end{enumerate}
The technique used for proving such statements relies on the well-known concept of \emph{shattering} and on the less known, dual concept of \emph{strong-shattering}. Shattering is commonly defined as a relation between $\powset{(\powset{(X)})}$ and $\powset{(X)}$, where $\powset{(X)}$ denotes the power set of $X$. Continuing the work of Litman and Moran~\cite{litman_moran,moran_thesis}, we present shattering as a relation between $\powset{(\{0,1\}^X)}$ and $\powset{(X)}$. This facilitates new definitions of shattering and strong-shattering that differ from each other only in the order of the quantifiers (see \S \ref{sec2}). The transposition of quantifiers demonstrates a certain duality between the two concepts of shattering. This duality enables an easy derivation of new, ``dual'' results from known results.


Our main tool is a generalization of the Sauer-Shelah Lemma which we term the \emph{Sandwich Theorem} (Theorem~\ref{thm:sandwich}) \cite{Pajor,BR95,Dress1,Greco98,ShatNews}. This theorem states that the size of a system $\s$ is at most the number of sets shattered by $\s$ and at least the number of sets strongly shattered by $\s$. Interpreting this theorem in the context of graph orientations yields inequalities that link orientations and subgraphs of the same graph. 


Systems for which the Sandwich Theorem collapses into an equality are called \emph{shattering-extremal} ($SE$). These systems were discovered independently several times by different groups of researchers~\cite{Law,BR95,Dress1,Dress2}. As far as we know, Lawrence \cite{Law} was the first to introduce them in his study of convex sets. Interestingly, the definition he gave does not require the concept of shattering. Independently, Bollob{\'a}s and Radcliffe \cite{BR95} discovered these systems, using the {\it shatters} relation (a.k.a.\ {\it traces}). Furthermore, they also introduced the relation of {\it strongly-shatters} (a.k.a.\ {\it strongly-traces}), and characterized shattering-extremal systems using the shatters and strongly-shatters relations. Dress et al.~\cite{Dress1,Dress2} discovered, independently of Bollob\'{a}s et al., the same characterization, and established the equivalence to the definition given by Lawrence. Several characterizations of these systems were given~\cite{Law,BR95,Dress2,Greco98,ShatNews,S-ext,moran_thesis}.

We present two general classes of $SE$ systems that stem from properties of graph orientations. One class is related to properties of distance in a weighted network, the other class is related to properties of flow networks. This allows us, in one direction, to apply known results about $SE$ systems to prove results concerning these two classes (such as (iv) above). In the other direction, the two classes form non-trivial clusters of new examples of $SE$ systems, and thus they may be useful for a better understanding of $SE$ systems. We note that the known characterizations of these systems are considered unsatisfactory by several authors \cite{BR95,Greco98,S-ext}, e.g.: ``...\ a structural description of extremal systems is still sorely lacking'' ~\cite{BR95}.

 
Some of the results presented in this paper can be proven in alternative ways. For example, McDiarmid~\cite{mcdiarmid81} implies some of our results (including (i) above) using general theorems from \emph{clutter percolation}.

\paragraph{Notational issue}
The systems we call shattering-extremal have been independently discovered several times in various contexts~\cite{Law,BR95}. Accordingly, such systems have been referred to as \emph{lopsided}, \emph{Sauer-extremal}, \emph{Pajor-extremal}, among other names. We call these systems \emph{shattering-extremal} ($SE$), following the work of Litman and Moran~\cite{litman_moran, moran_thesis}.


\newpage
\section{Preliminaries} \label{sec2}

In this section we introduce the concepts necessary to formulate and prove our results. The notation related to systems and shattering closely follows Litman and Moran~\cite{litman_moran,moran_thesis}.

\subsection{Systems} \label{subsec:systems}

In this paper, a \emph{system} is a pair $\system{S}{X}$, where $X$ is a set and $S\subseteq \{0,1\}^X$. Given a system ${\s=\system{S}{X}}$, we define the operators $S(\s) = S$, $C(\s) = \{0,1\}^X$, and $\dim(\s) = X$. For simplicity, we use $|\s|$ for $|S(\s)|$. A system $\s$ is \emph{trivial} if $S(\s) \in \{ \emptyset, C(\s)\}$. The  \emph{complement} of $\s$ is $\neg\s = \sys{C(\s)-S(\s)}{C(\s)}$. Note that there are exactly $2^{2^{\lvert X\rvert}}$ systems $\s$ for which ${{C(\s)=\{0,1\}^X}}$.

\subsection{Shattering and strong shattering} \label{subsec22}
Given two arbitrary functions  $f$ and $g$, and $A\subseteq \dom(f)\cap \dom(g)$, we say that $f$ \emph{agrees with} $g$ \emph{on} $A$, if $f(x)=g(x)$ for all $x\in A$. We say that $f$ \emph{agrees with} $g$, if they agree on the entire $\dom(f)\cap \dom(g)$.

We define shattering and strong shattering with the help of the \emph{merging} operator ($\mymerge$). Given two disjoint sets $X$ and $Y$, two functions $f\in \{0,1\}^X$ and $g\in \{0,1\}^Y$, let $f\mymerge g$ denote the unique function in $\{0,1\}^{X\cup Y}$ that agrees with both $f$ and $g$. Note that $\mymerge$ is a commutative and associative operator.

\begin{definition}
Let $\s$ be a system, let $X=\dim(\s)$ and let $Y \subseteq X$. We say that:

$\s$ \emph{shatters} $Y$, if: $$\left({\forall f\in\{0,1\}^Y}\right)\left({\exists g\in\{0,1\}^{X-Y}}\right): g\mymerge f\in S(\s).$$

$\s$ \emph{strongly shatters} $Y$, if: $$ \left({\exists g\in\{0,1\}^{X-Y}}\right)\left({\forall f\in\{0,1\}^Y}\right):g\mymerge f\in S(\s).$$

\end{definition}

Observe that the definitions of \emph{shatters} and \emph{strongly shatters} differ only in the order of the quantifiers. A straightforward application of predicate calculus gives the following result: 
\begin{lemma}\label{obs:dualStrSstr}
Let $\s$ be a system and let $\{X',X''\}$ be a partitioning of $\dim(\s)$. Then, exactly one of the following statements is true:
	\begin{enumpack}
	\item{$\s$ shatters $X'$}
	\item{$\neg\s$ strongly shatters $X''$}.
	~\hfill\qed
	\end{enumpack} 
\end{lemma} 

Two important subsets of $\powset(\dim(\s))$ are associated with $\s$: 

\begin{definition}
The \emph{shattered} and \emph{strongly shattered} sets of a system $\s$ are respectively:
\begin{align*}
\str(\s) & \defeq \setdef{Y\subseteq \dim(\s)}{Y \mbox{ is shattered by } \s}, \\
\sstr(\s) & \defeq \setdef{Y\subseteq \dim(\s)}{Y \mbox{ is strongly shattered by } \s}.
\end{align*}
\end{definition}

Clearly, both $\str(\s)$ and $\sstr(\s)$ are closed under the subset relation, and $\sstr(\s)\subseteq\str(\s)$.

Given a set $X$ and a family $\F\subseteq\powset(X)$, we define the \emph{co-complement} (${}^*$) operator as {{$\F^* = \{Y\subseteq X~:~Y^c\notin\F \}$}}. Observe that $^*$ is an involution\footnote{($\F^*)^*=\F$} and for any two families $\mathcal{A}$ and $\mathcal{B}$, we have $\mathcal{A} \subseteq \mathcal{B} \Leftrightarrow \mathcal{B}^* \subseteq \mathcal{A}^*$. 
With this operator, Lemma \ref{obs:dualStrSstr} can be expressed in the following way:
\begin{lemma}\label{lem:star str and sstr}
Let $\s$ be a system. Then
\begin{align*}
\str(\neg\s) &=  \sstr(\s)^*,\\
\sstr(\neg\s) &=  \str(\s)^*.
\end{align*}
\end{lemma}

The following theorem is the result of accumulated work by different authors, and parts of it were rediscovered independently several times (Pajor~\cite{Pajor}, Bollob{\'a}s and Radcliffe~\cite{BR95}, Dress~\cite{Dress1}, Holzman and Aharoni~\cite{ShatNews,Greco98}).

\begin{theorem}[Sandwich Theorem~\cite{Pajor,BR95,Dress1,ShatNews,Greco98}]\label{thm:sandwich}
For a system $\s$:
$$\lvert \sstr(\s)\rvert \leq \lvert\s\rvert \leq \lvert \str(\s)\rvert.$$
\end{theorem}
In the proof of Theorem~\ref{thm:sandwich}, given a system $\s$ and $x\in \dim(\s)$, we consider the following two ``sub-systems'' of $\s$, referred to as the \emph{restrictions of $\s$ associated with $x$}:
	\begin{align*} 
\sys{\{\left.f\right|_{\dim(\s)-\{x\}}~:~f\in\s,~f(x)=0\}}{\{0,1\}^{\dim(\s)-\{x\}}},\\
\sys{\{\left.f\right|_{\dim(\s)-\{x\}}~:~f\in\s,~f(x)=1\}}{\{0,1\}^{\dim(\s)-\{x\}}}.
	\end{align*}

\begin{proof}[of Theorem \ref{thm:sandwich}]
First, prove $\lvert\s\rvert \leq \lvert \str(\s)\rvert$.\\
Proceed by induction on $\dim(\s)$. The case $\dim(\s)=\emptyset$ is trivial. Otherwise, pick $x\in \dim(\s)$ and let $\s',\s''$ be the two restrictions of $\s$ associated with $x$. By the induction hypothesis, we have $\lvert\s'\rvert\leq\lvert\str(\s')\rvert$ and $\lvert\s''\rvert\leq\lvert\str(\s'')\rvert$. It is easy to verify that:
\vspace{-0.03in}
\begin{align*}
\setdef{Y\in \str(\s)}{x\notin Y} &\supseteq \setdef{Y}{Y\in\str(\s')\cup\str(\s'')},\\
\setdef{Y\in \str(\s)}{x\in Y} &\supseteq \setdef{Y\cup\{x\}}{Y\in\str(\s')\cap\str(\s'')}.
\end{align*}

Thus,
\begin{align*}
\lvert\str(\s)\rvert &= \lvert\setdef{Y\in \str(\s)}{x\notin Y}\rvert + \lvert\setdef{Y\in \str(\s)}{x\in Y}\rvert \\
							      &\geq \lvert\setdef{Y}{Y\in\str(\s')\cup\str(\s'')}\rvert\\
 &\quad\quad + \ \lvert\setdef{Y\cup\{x\}}{Y\in\str(\s')\cap\str(\s'')}\rvert \tag{by the above inclusions}\\
							      &= \lvert\str(\s')\cup\str(\s'')\rvert+\lvert\str(\s')\cap\str(\s'')\rvert\\
							      &= \lvert\str(\s')\rvert + \lvert\str(\s'')\rvert\\
							      &\geq \lvert\s'\rvert+\lvert\s''\rvert \tag{by the induction hypothesis}\\
							      &= \lvert\s\rvert.
\end{align*}

Next, prove $\lvert\s\rvert\geq\lvert\sstr(\s)\rvert$.\\
We use a certain duality between \emph{shattering} and \emph{strong shattering}. This duality manifests itself through a mechanical tranformation on text written in ``mathematical English''. It swaps the pair ``$\str$'' and ``$\sstr$'', the pair ``$\subseteq$'' and ``$\supseteq$'', and the pair ``$\leq$'' and ``$\geq$''. Note that the dual of ``$\lvert\s\rvert\leq\lvert\str(\s)\rvert$'' is  ``$\lvert\s\rvert\geq\lvert\sstr(\s)\rvert$''. It is easy to verify that the dual of the proof for $\lvert\s\rvert\leq\lvert\str(\s)\rvert$ is a valid proof for $\lvert\s\rvert\geq\lvert\sstr(\s)\rvert$. 
\end{proof}

It is important to note that the duality used in the proof is mysterious and fragile. It is not hard to find (true) claims whose dual claims are not true. For further discussion on dualities between shattering and strong shattering we refer the reader to Litman and Moran~\cite{litman_moran,moran_thesis}.
\begin{definition}
The \emph{VC-dimension} (Vapnik and Chervonenkis \cite{vapnik}) and the \emph{dual VC-dimension}~\cite{moran_thesis, litman_moran} of a system $\s$ are defined respectively as: 
\begin{eqnarray*}
\vc(\s) & \defeq & \max \setdef{|Y|}{Y\in \str(\s)}\footnotemark,\\
\dvc(\s) & \defeq & \max \setdef{|Y|}{Y\in \sstr(\s)}\footnotemark[\value{footnote}].
\end{eqnarray*}
\end{definition}
\footnotetext{As a special case, $\vc(\s)=\dvc(\s)=-1$ when $S(\s)=\emptyset$.}
Note that by the definition of the VC-dimension:
$$\str(\s)\subseteq\{Y\subseteq \dim(\s)~:~\lvert Y\rvert\leq \vc(\s) \}.$$
Hence, an easy consequence of Theorem \ref{thm:sandwich} is the following result:
\begin{theorem}[Sauer-Shelah Lemma~\cite{sauer,shelah}]\label{thm:Sauer lemma}
For a system $\s$ with $|\dim(\s)|=n$: $$|\s|\leq \sum_{i=0}^{\vc(\s)}{{n}\choose{i}}.$$
\end{theorem}

\subsection{Shattering-extremal systems} \label{sec41}

In this subsection we look at systems of a particular kind, namely those for which Theorem~\ref{thm:sandwich} collapses into an equality. We call these systems \emph{shattering-extremal}.

\begin{definition}
$\s$ is \emph{shattering-extremal} (in abbreviation: $SE$), if it satisfies $$\sstr(\s)=\str(\s).$$
\end{definition}

From Lemma~\ref{lem:star str and sstr} the following result is immediate:
\begin{lemma}\label{lem:complement}
Let $\s$ be a system. Then 
\begin{equation}
\s \mbox{ is } SE \ \Longleftrightarrow \   \neg \s \mbox{ is } SE. \quad{\qed}\notag
\end{equation}
\end{lemma}

Similarly to Theorem~\ref{thm:sandwich}, the following result has also been rediscovered independently several times (Bollob{\'a}s and Radcliffe~\cite{BR95}, Dress et al~\cite{Dress2}). 
\begin{theorem}[\cite{BR95,Dress2}]\label{thm:SE char}
Let $\s$ be a system. The following statements are equivalent:
	\begin{enumpack}
	\item{$\s$ is $SE$}
	\item{$\lvert \sstr(\s)\rvert=\lvert \s\rvert$}
	\item{$\lvert \s\rvert=\lvert \str(\s)\rvert. \mbox{\quad\qquad\qed}$} 
	\end{enumpack}
\end{theorem}

\section{Systems of orientations} \label{sec3}


We use standard terminology of graph theory (cf.\ Bollob\'{a}s~\cite{bollobas_book}). Given an undirected graph $G$ with vertex set $V(G)$ (or simply $V$) and edge set $E(G)$ (or simply $E$), let $n=|V|$ and $m=|E|$. In this work, a \emph{subgraph of} $G$ is a graph $G'$ with $V(G')=V(G)$ and $E(G')\subseteq E(G)$. For $X\subseteq E(G)$ we denote by $G_{X}$ the subgraph of $G$ with $E(G_{X})=X$. We consider only simple graphs, however, the results can easily be extended to non-simple graphs.

An \emph{orientation} of a graph $G$ is an assignment of a direction to each edge. To encode such an assignment as a function ${{d:E \rightarrow \{0,1\}}}$, choose a \emph{canonical orientation} $\vec{E}$ and interpret $d$ relative to $\vec{E}$ in the obvious way: $d$ \emph{orients} an edge $e \in E$ in agreement with $\vec{E}$ if $d(e)=0$, and opposing $\vec{E}$ if $d(e)=1$. For an orientation $d$, let $\vec G^d$ denote the digraph obtained by orienting the edges of $G$ according to $d$ and to some canonical orientation. For simplicity, we also refer to the digraph $\vec G^d$ as an \emph{orientation} of $G$. The set of all orientations of $G$ is $\O(G) = \{0,1\}^{E}$. A \emph{system of orientations} of $G$ is a system $\s$, where $C(\s) = \O(G)$.  


\subsection{Cycles and forests}\label{sec31}
As a warm-up, and to illustrate our techniques, we consider the system of all cyclic orientations of a graph $G$, denoted as $\s_{cyc}(G) = \systemexp{\setdef{d \in \O(G)}{\vec G^d \ \mbox{has a directed cycle}}}{\O(G)}$ and its complement, $\neg\s_{cyc}$, namely the system of all acyclic orientations. We prove the following inequalities:

\begin{theorem}\label{thm: cyclic a-cyclic, sstr str}
Let $G$ be a graph. Then:
\begin{enumpack}
\item \begin{tabular}{M{0.43}M{.01}M{0.45}}
the number of orientations of $G$ that contain a directed cycle & $\geq$ & the number of subgraphs of $G$ that contain an undirected cycle,
\end{tabular}
\item \begin{tabular}{M{0.43}M{.01}M{0.45}}
the number of acyclic orientations of $G$ & $\leq$ & the number of subgraphs of $G$ that are forests.
\end{tabular}
\end{enumpack}
\end{theorem}

Note that the two inequalities are equivalent. However, proving them in parallel illustrates a certain duality (symmetry).
To derive these inequalities, we characterize $\sstr(\s_{cyc})$ and $\str(\neg\s_{cyc})$.

\begin{lemma}\label{lem: cyclic, a-cyclic - sstr, str}
Let $G$ be a graph. Then:
\begin{enumpack}
\item{$\{X\subseteq E~:~G_{E-X} \ \mbox{has a cycle} \}=\sstr(\s_{cyc})$.}
\item{$\{X\subseteq E~:~G_{X} \ \mbox{is a forest} \}=\str(\neg\s_{cyc})$.}
\end{enumpack} 
\end{lemma}
Note that by Lemma~$\ref{obs:dualStrSstr}$ the two statements are equivalent. To establish Lemma~\ref{lem: cyclic, a-cyclic - sstr, str}, we prove:
\begin{lemma}\label{lem: helper cyclic, a-cyclic}
Let $G$ be a graph. Then:
\begin{enumpack}
\item{$\{X\subseteq E~:~G_{E-X} \ \mbox{has a cycle} \}\subseteq\sstr(\s_{cyc})$.}
\item{$\{X\subseteq E~:~G_{X} \ \mbox{is a forest} \}\subseteq\str(\neg\s_{cyc})$.}
\end{enumpack} 
\end{lemma}
\vspace{-10pt}
\begin{proof}
\textit{(i)} If there exists a cycle $C \subseteq E - X$, then there exists an orientation $d$ of $E - X$, such that $C$ becomes a directed cycle. Clearly, every extension of $d$ to an orientation of $E$ contains this directed cycle. This means that $X \in \sstr(\s_{cyc})$.

\textit{(ii)} Assume that $G_{X}$ is a forest. We need to show that every orientation $d$ of $X$ can be extended to an acyclic orientation of $E$. Since $G_X$ is a forest, $\vec G_{X}^{d}$ is a DAG whose edges form a pre-order $P$ on $V$. Pick (by topological sorting) a linear order $L$ of $V$ that extends $P$ and orient the edges of $E-X$ according to $L$ (from smaller to larger vertex). Clearly, the resulting orientation is acyclic.
\end{proof}

Applying the co-complement operator to both sides of the equations of Lemma~\ref{lem: helper cyclic, a-cyclic} gives:
\begin{lemma}\label{lem:dual}
Let $G$ be a graph. Then:
\begin{enumpack}
\item{$\{X\subseteq E~:~G_{X} \ \mbox{is a forest} \}\supseteq\str(\neg\s_{cyc})$.}
\item{$\{X\subseteq E~:~G_{E-X} \ \mbox{has a cycle} \}\supseteq\sstr(\s_{cyc})$.} ~\hfill\qed
\end{enumpack} 
\end{lemma}
Lemma~\ref{lem: helper cyclic, a-cyclic} and Lemma~\ref{lem:dual} together imply Lemma~\ref{lem: cyclic, a-cyclic - sstr, str}. An application of the Sandwich Theorem (Theorem~\ref{thm:sandwich}) yields the inequalities of Theorem \ref{thm: cyclic a-cyclic, sstr str}. Moreover, from the characterizations of $\str(\neg \s_{cyc})$ and $\sstr(\s_{cyc})$ it follows that:  
\begin{proposition}\label{prop1} 
Let $\s_{cyc}$ denote the system of all cyclic orientations of $G$, then: \\
\vspace{-0.2in}
\begin{enumpack}
\item{$\vc(\neg \s_{cyc}) = n-k$, where $k$ is the number of connected components of $G$.}
\item{$\dvc(\s_{cyc}) = m-c$, where $c$ is the size of the smallest cycle in $G$.}
\end{enumpack}
\end{proposition}
An application of Theorem~\ref{thm:Sauer lemma} and a standard bound on binomial sums yield:
\begin{proposition}\label{prop2}
The number of acyclic orientations of $G$ is at most ${\left( \frac{m \, \euler}{n-k} \right)}^{n-k}$ where $k$ is the number of connected components of $G$.
\end{proposition}

In general, the inequalities of Theorem~\ref{thm: cyclic a-cyclic, sstr str} are strict. In fact, as implied by Theorem~\ref{thm:lopsidedness} presented in \S\,\ref{sec4}, these inequalities are strict if and only if $G$ contains a cycle. Note that the second statement of Theorem~\ref{thm: cyclic a-cyclic, sstr str} appears to be known (it is implied by an identity of Bernardi~\cite{bernardi}). Note as well that the number of acyclic orientations and the number of subgraphs that are forests are two particular values of the Tutte-polynomial~\cite{bollobas_book}.  Aharoni and Holzman~\cite{aharoni_holzman} have brought to our attention that the result can also be proven by induction, using the graph operations of \emph{deletion} and \emph{contraction} of edges. 

A natural question is whether $\str(\s_{cyc})$ and $\sstr(\neg \s_{cyc})$ can be similarly described. The characterization of these sets, however, seems to be less natural. Nevertheless, we can observe the following:

\begin{lemma}\label{lem:bridges}
Let $G$ be a graph. Then:
\begin{enumpack}
\item{$ \{X\subseteq E~:~E-X\ \mbox{intersects a cycle of}\ G\} \supseteq \str(\s_{cyc})$.}
\item{$ \{X\subseteq E~:~X\ \mbox{contains only bridges}\} \subseteq \sstr(\neg\s_{cyc})$.}
\end{enumpack} 
\end{lemma} 
\vspace{-10pt}
\begin{proof}
By Lemma~\ref{lem:star str and sstr} and the monotonicity of the co-complement operator, the two statements are equivalent. Thus, it is sufficient to prove \textit{(ii)}:
Let $d$ be an acyclic orientation of $E-X$. If $X$ contains only bridges, every extension of $d$ to $E$ remains acyclic (by definition, a bridge is not contained in any cycle). Therefore, $X$ is strongly shattered by $\neg \s_{cyc}$.
\end{proof}

Similarly to the proof of Theorem \ref{thm: cyclic a-cyclic, sstr str}, Lemma~\ref{lem:bridges}, together with the Sandwich Theorem yield certain inequalities. One can also derive an upper bound for $\vc(\s_{cyc})$ and a lower bound for $\dvc(\neg \s_{cyc})$. 

\subsection{Strong orientations}\label{sec32}
Let $k\in\mathbb{N}$. A graph $G$ is $k$-\emph{edge-connected} if it remains connected whenever fewer than $k$ edges are removed. A digraph $\vec G$ is $k$-\emph{arc-strong} if for every $u,v\in V(\vec G)$ there exist $k$ edge-disjoint paths from $u$ to $v$. Since in this section we only refer to edge-connectivity, we use the shorter terms $k$-\emph{connected}, $k$-\emph{strong}, and \emph{disjoint}.

\begin{theorem}\label{thm:strong}
For an arbitrary graph $G$:

\begin{tabular}{M{.25}M{.005}M{.29}M{.005}M{.288}}
the number of $2k$-connected subgraphs of $G$ & $ \leq $ & the number of $k$-strong orientations of $G$ & $ \leq $ & the number of $k$-connected subgraphs of $G$.
\end{tabular}

\end{theorem}

Before proving this theorem, we state two well-known results. The first is a characterization of graphs that admit a $k$-strong orientation. The second is an immediate consequence of Menger's theorem for directed graphs.

\begin{theorem}[Nash-Williams~\cite{nash-williams}]\label{thm:nash-williams}
A graph $G$ has a $k$-strong orientation iff $G$ is $2k$-connected.~\hfill\qed
\end{theorem}

\begin{theorem}[Menger~\cite{bollobas_book}]\label{stronglemma}
A digraph $\vec G$ is $k$-strong iff every non-trivial cut of $\vec G$ contains at least $k$ forward edges. ~\hfill\qed
\end{theorem}

\begin{proof}[of Theorem~\ref{thm:strong}]
For $i\in\mathbb{N}$, let $\s_{i}$ denote the system of $i$-strong orientations of $G$ and let $\F_i$ denote the family of sets $X \subseteq E$, such that $G_X$ is $i$-connected.

\noindent $\lvert \F_{2k}\rvert\leq \lvert\s_{k}\rvert$: \quad
By the Sandwich Theorem, $\lvert \sstr(\s_k) \rvert\leq\lvert \s_k \rvert$. Thus, it is sufficient to show that
$$\{X\subseteq E~:~G_{E-X}\ \mbox{is $2k$-connected}\}\subseteq\sstr(\s_{k}).$$
Indeed, if $G_{E-X}$ is $2k$-connected, then by Theorem~\ref{thm:nash-williams}, there exists an orientation $d$ of $E-X$, such that $\vec G^d$ is $k$-strong. Clearly, all extensions of $d$ to an orientation of $E$ maintain the $k$-strong property, and thus $X\in\sstr(\s_{k})$.

\noindent $\lvert\s_{k}\rvert\leq\lvert \F_{k}\rvert $:\quad
By the Sandwich Theorem, $\lvert \s_k \rvert\leq\lvert \str(\s_k) \rvert$. Thus, it is sufficient to show that
$$\str(\s_{k})\subseteq\{X\subseteq E~:~G_{E-X}\ \mbox{is $k$-connected}\}.$$
Let $X\in\str(\s_{k})$. It is enough to show that every non-trivial cut of $G$ contains at least $k$ edges in $E-X$. Let $(V', V'')$ be a non-trivial cut of $G$. Pick an orientation $d$ of $X$ that directs every edge included in the cut as a \emph{backward} edge. Since $X\in\str(\s_{k})$, there exists a $k$-strong orientation $f$ of $E$ which extends $d$. By Theorem~\ref{stronglemma}, the cut $(V',V'')$ must contain at least $k$ \emph{forward} edges in $\vec G^f$ and by the choice of $d$, all of these edges are from $E-X$.
\end{proof}
Remark: After discovering the above result, it was brought to our attention that Mcdiarmid~\cite{mcdiarmid81} has proved a similar result using non-trivial tools from clutter percolation. 

In general, the inequalities of Theorem~\ref{thm:strong} are strict. In fact, as implied by Theorem~\ref{thm:lopsidedness} presented in \S\,\ref{sec4}, both of these inequalities are strict if and only if $G$ is $2k$-connected.

Let $c_k(G)$ denote the size of the \emph{minimum $k$-connected subgraph} of $G$. Computing $c_k$ is a known NP-hard problem~\cite{garey-johnson}, even for $k=2$. From the proof of Theorem~\ref{thm:strong} and from Theorem~\ref{thm:Sauer lemma} we obtain:

\begin{proposition}\label{prop3} If $\s_{k}$ is the system of all $k$-connected orientations of $G$, then: 
$$ m-c_{2k} \leq \dvc(\s_{k}) \leq \vc(\s_{k}) \leq m-c_k.$$
\end{proposition}

\begin{proposition}\label{prop4}
The number of $k$-connected orientations of $G$ is at most ${\left( \frac{m \, \euler}{m-c_k} \right)}^{m-c_k}$.
\end{proposition}

Observe that a simple corollary of Theorem~\ref{thm:strong}, together with classical results on the connectivity of random graphs in the $G(n,p)$ model~\cite{Walsh} is that almost every tournament on $n$ vertices is $k$-strong, for any fixed positive $k$.

\subsection{General inequality}

We can abstract away parts of the earlier proofs, to obtain the following result:
  
\begin{theorem}[General inequality]\label{thm:meta-ineq}
Let $P$ be a monotone increasing property of digraphs, and $P'$ be a property of graphs, such that if a graph $G$ satisfies $P'$, then there exists an orientation of $G$ that satisfies $P$. Then, for an arbitrary graph $G$:

\begin{tabular}{M{0.35}M{.015}M{0.35}}
the number of subgraphs of $G$ that satisfy $P'$ & $\leq$ & the number of orientations of $G$ that satisfy $P$.
\end{tabular}
\end{theorem}
\begin{proof}
Let $\s_P$ be the system of orientations of $G$ that satisfy $P$. Let $X \subseteq E$, such that $P'(G_{E-X})$ holds. From the conditions it follows that there exists an orientation $d$ of $G_{E-X}$, such that $P(\vec G_{E-X}^d)$ holds. Since $P$ is monotone increasing, $P$ holds for any extension of $d$ to $E$. It follows that $X \in \sstr(\s_P)$. Hence, $\setdef{X \subseteq E}{P'(G_{E-X})} \subseteq \sstr(\s_P)$, and Theorem~\ref{thm:sandwich} yields the result.
\end{proof}

\subsection{Further inequalities}

The conditions of Theorem~\ref{thm:meta-ineq} are fulfilled by many natural connectivity-properties of digraphs. These include $s $ - $ t$ connectivity, rootedness, unilateral connectivity, or the existence of a Hamiltonian cycle.

As a further application, consider the following problem: Given a graph $G$ and a digraph $\vec{H}$, denote by $D(G,\vec{H})$ the number of
orientations of G not containing a copy of $\vec{H}$. Erd\H{o}s~\cite{erdosopen} posed the question of estimating $D(G,\vec{H})$, and researchers have studied many variants of this problem~\cite{erdos86,alon06,kmp2011}. Let $H$ denote the undirected graph that underlies $\vec{H}$ and let $D'(G,H)$ denote the number of subgraphs of $G$ not containing a copy of $H$. Then, similarly to the preceeding results, we obtain:
\begin{theorem}\label{thm:forbidden orientation}
$D(G,\vec{H})\leq D'(G,H)$. ~\hfill\qed
\end{theorem}


Using the same approach as previously, we obtain that the VC-dimension of the system of orientations not containing a copy of $\vec{H}$, is at most $ex(G,H)$, the size of the largest subgraph of $G$ which does not contain a copy of $H$. When $G = K_n$ (the complete graph on $n$ vertices) the quantity has been denoted as $ex(n,H)$. When $H$ is also a complete graph, this is the well-known Tur{\'a}n number~\cite{erdos86}. 

A result of Erd\H{o}s, Frankl, and R\"{o}dl~\cite{erdos86} states that the number of graphs on $n$ vertices that do not contain a copy of $H$ is $2^{{ex(n,H)}(1 + o(1))}$, provided that the chromatic number of $H$, $\chi(H) \geq 3$. This result, together with Theorem~\ref{thm:forbidden orientation} yield:

\begin{proposition}\label{erdosprop}
Let $H$ be a graph with $\chi(H)\geq 3$ and let $\vec H$ be an orientation of $H$. 
Then: 
$$D(K_n,\vec H) \leq 2^{{ex(n,H)}(1 + o(1))}.$$
\end{proposition}

When $\vec H$ is a tournament, Proposition~\ref{erdosprop} is implied by a result of Alon and Yuster~\cite{alon06}, proven using sophisticated techniques.

\section{Shattering-extremal systems of orientations} \label{sec4}

In this section we present $SE$ systems of orientations, or equivalently, systems for which the Sandwich Theorem collapses into an equality. Two natural and general classes of $SE$ systems correspond to orientations with a certain $s $ - $ t$ flow, respectively $s $ - $ t$ distance in a graph. Many other results, including Theorem~\ref{thm:steiner}, are direct consequences of the results for the flow or distance examples. The selection of examples is not exhaustive and there exist natural classes of $SE$ systems of orientations that seem not to be reducible to either the flow or the distance example.

\subsection{Flow} \label{flow}
We consider flow in both directed and undirected graphs (in an undirected graph an edge can be used in both directions). Let $G$ be a graph, let $c:E(G)\rightarrow \mathbb{R}^{\geq 0}$ be a \emph{capacity} function of the edges, $s\in V(G)$ be the \emph{source} and $t\in V(G)$ be the \emph{sink}. For a number $w\in\mathbb{R}$, let $\s_w$ denote the system of all orientations of $G$ for which there exists a \emph{flow} of size (at least) $w$ from $s$ to $d$. We show that:
\begin{theorem}\label{thm:flow SE}
$\s_w$ is $SE$.
\end{theorem}
As a corollary of the proof of Theorem \ref{thm:flow SE} we obtain the following identity:
\begin{theorem}\label{thm:flow equality}
For an arbitrary graph G:

\begin{tabular}{M{0.45}M{.015}M{0.43}}
the number of orientations of $G$ for which there exists a flow of size $w$ & $=$ & the number of subgraphs of $G$ for which there exists a flow of size $w$.
\end{tabular}
\end{theorem}

The following well-known equivalence is useful in establishing Theorem \ref{thm:flow SE}:
\begin{theorem}[max-flow min-cut]\label{thm:cut-flow}
The following two statements are equivalent:
\begin{enumpack}
\item \emph{There exists} a flow of size $w$.
\item \emph{Every} $s $ - $ t$ cut has capacity of at least $w$. ~\hfill\qed
\end{enumpack}
\end{theorem}
\begin{lemma}\label{lem:flow1}\ \\
$ \{X\subseteq E~:~\mbox{There exists a flow of size}\ w\ \mbox{in}\ G_{E-X}\}\subseteq\sstr(\s_w)$
\end{lemma}
\begin{proof}
Let $X\subseteq E$ be such that there exists a flow of size $w$ in $G_{E-X}$ and let $f$ be such a flow. Assume w.l.o.g.\ that $f$ is acyclic. Pick an orientation $d$ of $E-X$ such that every edge with positive flow is oriented in the direction of the flow. In $\vec G_{E-X}^d$, $f$ is a flow of size $w$. The flow $f$ remains feasible in every extension of $d$ to $E$. This means that $X\in \sstr(\s_w)$.
\end{proof}
\begin{lemma}\label{lem:flow2}\ \\
$ \str(\s_w)\subseteq\{X\subseteq E~:~\mbox{There exists a flow of size}\ w\ \mbox{in}\ G_{E-X}\}$
\end{lemma}
\begin{proof}
Let $X\in \str(\s_w)$. It is enough to show that every $s $ - $ t$ cut has capacity of at least $w$ in $G_{E-X}$. Let $(S,T)$ be an $s $ - $ t$ cut. Choose an orientation $d$ of $X$ that orients all edges of $X$ which are contained in the cut as \emph{backward} edges. Since $X\in\str(\s_w)$, there exists an extension of $d$ in which there is a flow of size $w$. Thus, by Theorem \ref{thm:cut-flow}, in the resulting digraph the capacity of $(S,T)$ is at least $w$. The claim follows from the fact that all forward edges of the cut are from $E-X$.
\end{proof}
From Lemmas \ref{lem:flow1} and \ref{lem:flow2} it follows that $\str(\s_w)\subseteq\sstr(\s_w)$. Theorem \ref{thm:flow SE} follows from the fact that the reverse inclusion always holds. Theorem \ref{thm:flow equality} follows by the Sandwich Theorem, combined with Lemmas \ref{lem:flow1} and \ref{lem:flow2}.

Lemmas \ref{lem:flow1} and \ref{lem:flow2} give a characterization of $\vc(\s_w)$. For $w\in\mathbb{R}$, let $e_w$ denote the size\footnote{As a special case, $e_w=-1$ when there exists no flow of size $w$ in $G$.} of the smallest subgraph of $G$ that admits an $s $ - $ t$ flow of size $w$. Computing $e_w$ is a natural NP-hard optimization problem (it reduces to minimum Steiner tree, as shown in \S\,\ref{furtherexamples}).

\begin{proposition}\label{prop:vc flow}
$\dvc(\s_w)=\vc(\s_w)=m-e_w$.
\end{proposition}

From Lemma~\ref{lem:complement} it follows that $\neg\s_w$, namely the system of orientations with a maximum flow less than $w$, is also $SE$. An application of Lemma \ref{lem:star str and sstr} yields a characterization of $\str(\neg\s_w)$ $(\,=\sstr(\neg \s_w)\,)$ and results analogous to Theorem \ref{thm:flow equality} and Proposition \ref{prop:vc flow} regarding $\neg\s_w$.

\subsection{Distance} \label{distance}
We consider distance in both directed and undirected graphs (in an undirected graph an edge can be used in both directions). Let $G$ be a graph, let $w:E(G)\rightarrow \mathbb{R}^{\geq 0}$ be a \emph{length} function of the edges, $s\in V(G)$ be the \emph{source} and $t\in V(G)$ be the \emph{destination}. For a number $d\in\mathbb{R}$, let $\s_d$ denote the system of all orientations of $G$ in which the distance from $s$ to $t$ is at most $d$. We show that:
\begin{theorem}\label{thm:dist SE}
$\s_d$ is $SE$.
\end{theorem}
As a corollary of the proof of Theorem \ref{thm:flow SE} we obtain the following identity:
\begin{theorem}\label{thm:dist equality}
For an arbitrary graph G:

\begin{tabular}{M{0.45}M{.01}M{0.45}}
the number of orientations of $G$ in which the distance from $s$ to $t$ is at most $d$ & $=$ & the number of subgraphs of $G$ in which the distance from $s$ to $t$ is at most $d$.
\end{tabular}
\end{theorem}

In order to highlight a certain symmetry between flow and distance (or more precisely, between minimum cut size and minimum path length), we establish the results of this subsection in a manner analogous to the proofs of \S\,\ref{flow}.

We call $\pi:V(G) \rightarrow \mathbb{R}^{\geq 0}$ a \emph{potential function} of the vertices, if for every edge $(u,v)$ the condition $\pi(v) - \pi(u) \leq w(u,v)$ holds. The \emph{potential difference} of $G$ with respect to $\pi$ is $\pi(t)-\pi(s)$. The following easily verifiable equivalence helps in establishing Theorem \ref{thm:dist SE}:
\begin{theorem}\label{thm:potential}
The following two statements are equivalent:
\begin{enumpack}
\item \emph{There exists} a potential function with potential difference $d$.
\item \emph{Every} $s $ - $ t$ path has length at least $d$. ~\hfill\qed
\end{enumpack}
\end{theorem}

\begin{lemma}\label{lem:dist1}\ \\
$ \{X\subseteq E~:~\mbox{The distance from $s$ to $t$ in $G_{E-X}$ is at most $d$}\}\subseteq\sstr(\s_d)$
\end{lemma}
\begin{proof}
Let $X\subseteq E$ be such that there exists an $s $ - $ t$ path of length at most $d$ in $G_{E-X}$ and let $p$ be such a path. Pick an orientation $f$ of $E-X$ such that in $\vec G_{E-X}^f$, the path $p$ is oriented from $s$ to $t$. In every extension of $f$ to $E$, the path $p$ remains a valid $s $ - $ t$ path of length at most $d$. This means that $X\in \sstr(\s_d)$.
\end{proof}
\begin{lemma}\label{lem:dist2}\ \\
$ \str(\s_d)\subseteq\{X\subseteq E~:~\mbox{The distance from $s$ to $t$ in $G_{E-X}$ is}\\ \mbox{ at most $d$}\}$
\end{lemma}
\begin{proof}
Let $X\in \str(\s_d)$. We show that every potential function for $G_{E-X}$ gives rise to a potential difference of at least $d$. Let $\pi$ be a potential function for $G_{E-X}$. Choose an orientation $f$ of $X$ that orients all edges of $X$ from larger towards smaller potential, according to $\pi$. Clearly, $\pi$ remains a valid potential function for the entire $G$. Since $X\in\str(\s_d)$, there exists an extension of $f$ in which there is an $s $ - $ t$ path of length at most $d$. Thus, by the complement of Theorem \ref{thm:potential}, in the resulting digraph, the potential difference is at least $d$. The claim follows from the fact that all potential-increasing edges are from $E-X$.
\end{proof}
From Lemmas \ref{lem:dist1} and \ref{lem:dist2} it follows that $\str(\s_d)\subseteq\sstr(\s_d)$. Theorem \ref{thm:dist SE} follows from the fact that the reverse inclusion always holds. Theorem \ref{thm:dist equality} follows by the Sandwich Theorem, combined with Lemmas \ref{lem:dist1} and \ref{lem:dist2}.

Let $d\in\mathbb{R}$, let $p_d$ denote the size\footnote{As a special case, $p_d=-1$ when there is no path of length at most $d$ in $G$.} of the smallest subgraph of $G$ that contains an $s $ - $ t$ path of length at most $d$. Observe that $p_d$ can be computed in polynomial time. 

\begin{proposition}\label{prop:vc dist}
$\dvc(\s_d)=\vc(\s_d)=m-p_d$.
\end{proposition}

From Lemma~\ref{lem:complement} it follows that $\neg\s_d$, namely the system of orientations with $s $ - $ t$ distance more than $d$, is also $SE$. Again, we obtain results analogous to Theorem \ref{thm:dist equality} and Proposition \ref{prop:vc dist} regarding $\neg\s_d$.

\subsection{Further examples}\label{furtherexamples}

Many ``natural'' systems of orientations can be viewed as special cases of the above systems. 

As a first example, let $G$ be a graph, let $s\in V(G)$ and $W\subseteq V(G)$. Let $\s_{s,W}(G)$ denote the system of all orientations for which every $w\in W$ is reachable from $s$. Transform $G$ into a flow network $G'$ by letting the capacities of all $e\in E(G)$ be infinity, designating $s$ as the source, and adding a destination $t$ which is connected to every $w\in W$ with edges of unit capacity. It is not hard to see that $\s_{s,W}(G)$ is transformed into $\s_{\left\lvert W\right\rvert}(G')$. In this case, Theorem \ref{thm:flow equality} and Proposition \ref{prop:vc flow} give the following results:

\begin{theorem}\label{thm:steiner} For an arbitrary graph $G$:

\begin{tabular}{M{0.46}M{.01}M{0.43}}
the number of orientations of $G$ in which every $w\in W$ is reachable from $s$ & $=$ & the number of subgraphs of $G$ in which $W\cup\{s\}$ is connected.
\end{tabular}
\end{theorem}

\begin{proposition}
Let $t$ be the size of a minimum unweighted Steiner tree for $W\cup\{s\}$. Then:
$$\dvc(\s)=\vc(\neg\s)=m-t. $$
\end{proposition}

A result equivalent to Theorem~\ref{thm:steiner} was proven recently by Linusson~\cite{linusson}. Note that ``$\s_{s,W}$ is $SE$'' can be proven directly, without proving the more general flow-result first. As special cases, when $W = \{t\}$, the system $\s_{s,W}$ consists of all orientations with a path from $s$ to $t$, and when $W = V-\{s\}$, the system $\s_{s,W}$ consists of all orientations in which $s$ is a root. We obtain equalities between the number of orientations that admit an $s $ - $ t$ path and the number of subgraphs in which $s$ and $t$ are connected, respectively between the number of orientations in which $s$ is a root and the number of connected subgraphs. 

We present another example: let $G$ be a graph with edge lengths $w:E(G) \rightarrow \mathbb{R}^{\geq 0}$, let $A\subseteq V(G)$ and $B\subseteq V(G)$. Let $\s_{A,B,d}(G)$ denote the system of all orientations for which there exist $u\in A$ and $v\in B$, such that the distance from $u$ to $v$ is at most $d$. The following transformation can be made: add source $s$ and destination $t$ to $G$, and connect every vertex in $A$ to $s$, respectively every vertex in $B$ to $t$, using edges of zero length. Denoting the obtained graph by $G'$, we can see that $\s_{A,B,d}(G)$ is transformed into $\s_{d}(G')$. The results are analogous to Theorem~\ref{thm:dist equality} and Proposition~\ref{prop:vc dist}.

\ignore{
\lk{for the record: we did not include the following SE examples:\\
* $A-to-B$ flow $\geq w_i$ (reduce to flow) $w_i$ can be parametrized either by $|A|$ or by $|B|$\\
* $A-to-B$ distance $\leq d_i$ (reduce to distance) $d_i$ can be parametrized either by $|A|$ or by $|B|$\\
* $s$ reaches all $W$, distances $\leq d_i$ (not reducible (?))\\
* $s$ sends flow $\geq w_i$ to at least one of $t_i$ (SE?)
}

\lk{can we say something intelligent about relation to LP duality?} \lk{see wiki article "LP-type problem" - is there some similar characterization here?}
}
\subsection{Discussion} \label{discussion}


Several characterizations of $SE$ systems were given in the past twenty years~\cite{BR95,Dress2,Greco98,moran_thesis}. In this subsection, we present a characterization which was found independently by Lawrence~\cite{Law}, and by Bollob{\'a}s and Radcliffe~\cite{BR95} and seems to be more natural in the context of graphs. To formalize this characterization, we need to introduce some concepts. 

Let $X$ be a set. Given $Y\subseteq X$, a $Y$-\emph{cube} of $\{0,1\}^X$ is an equivalence class of the following equivalence relation on $\{0,1\}^X$: \ \  \framebox{$u\sim v$: ``$u$ agrees with $v$ on $X-Y$''}. A \emph{cube} of $\{0,1\}^X$ is a $Y$-cube for some $Y\subseteq X$. Given $C$, a $Y$-cube of $\{0,1\}^X$, we define $\dim(C) = Y$. Note that the number of $Y$-cubes is $2^{|X-Y|}$, they are mutually disjoint and they cover $\{0,1\}^X$.

Let ${\s=\system{S}{X}}$ be a system and $C \subseteq C(\s)$ be a cube. It is useful to consider the structure $\systemexp{S(\s) \cap C}{C}$ as a system. However, formally it is not a system. To deal with this technicality we give the following definition: The \emph{restriction} of $\s$ to $C$ is the system ${\system{Q}{\dim(C)}}$ where $Q = \setdef{\left.f\right|_{\dim(C)} }{f\in S\cap C}$. A system is a \emph{restriction} of $\s$ if it is a restriction of $\s$ to some cube $C\subseteq C(\s)$. 

The \emph{antipodal} system of $\s$ is $\ap(\s) = \system{\bar{S}}{X}$, where $\bar{S} = \setdef{f}{1-f \in S}$. A system $\s$ is \emph{symmetric}, if $\s = \ap(\s)$. Recall that a system $\s$ is trivial if $S(\s)\in\{\emptyset,C(\s)\}$.

\begin{theorem}[Lopsidedness~\cite{Law,BR95}]\label{thm:lopsidedness}
A system $\s$ is $SE$, iff it has no non-trivial, symmetric restrictions.  ~\hfill\qed
\end{theorem}

Note that the systems discussed in \S\,\ref{sec31} and \S\,\ref{sec32} are symmetric and thus they are $SE$ if and only if they are trivial. For example, consider $\s_{cyc}$, the system of all cyclic orientations of $G$, and let $Y\subseteq E(G)$. A $Y$-cube $C$ of $\O(G)$ corresponds to a partial orientation of $G$ in which only the edges of $E(G)-Y$ are oriented. The restriction of $\s$ to $C$ corresponds to all extensions of the partial orientation to an orientation of $G$ which contains a cycle. In this case, the system $\s_{cyc}$ is symmetric, since flipping all the edges of a directed cycle yields a directed cycle. Also, $\s_{cyc}$ is trivial if and only if $G$ is a forest.
 
Let us define the \emph{flip-distance} between two orientations $f$ and $g$ of a graph $G$ as the number of edges in $E(G)$ on which $f$ and $g$ differ. Let $P$ be a property of orientations, such that the corresponding system, $\s_P$, is $SE$. From a known property of $SE$ systems~\cite{BR95, Dress2}, it follows that the orientations in $\s_P$ form an isometric subgraph of the hypercube $\{0,1\}^{E}$, with edges connecting pairs of orientations of flip-distance one (\emph{partial cube} property). As a consequence, if $f$ and $g$ are orientations of $G$ satisfying $P$, then there exists a sequence of edge-flips from $f$ to $g$, with all intermediate orientations of the sequence satisfying property $P$ and with the length of the sequence equal to the flip-distance between $f$ and $g$. 

\section{Conclusion and further work}

In this paper we studied a variety of natural properties of graph orientations. In particular, we have shown that for many of these properties, concepts related to VC-theory, such as VC-dimension and shattering have natural interpretations. 

One natural question is whether the graph-theoretical results presented in this paper can be proven more directly. In particular: do there exist natural injective and surjective maps that imply the different inequalities? 

Another possible application of the connection presented in this paper is to a particular type of supervised learning problem, that could be called ``orientation learning''. In this type of problem, a graph $G$ and a property $P$ of orientations are given. A single target orientation $d$ satisfying $P$ is selected to be learned. In machine learning terms, $P$ is the \emph{hypothesis space}, and $d$ is the \emph{target hypothesis}. The problem can be formulated both in the classical \emph{passive learning} or in the \emph{active learning}~\cite{settles} framework. 

Different variants of these problems have been studied, mostly for complete and random graphs, where $P$ is the property of acyclic orientations~\cite{alontuza,HuangKK11,Ailon12}. We believe that the connections presented in this paper may be useful for understanding the learnability of other properties $P$ in other classes of graphs as well.

\section*{Acknowledgements} \label{ack}
This work benefited from discussions with Nir Ailon, Ron Holzman, Ami Litman, Kurt Mehlhorn, Shlomo Moran, Rom Pinchasi, and Raimund Seidel. The work developed from results in the Master thesis of Shay Moran; we would like to acknowledge again the contribution of the advisor of this thesis - Ami Litman.

\newpage

\bibliographystyle{abbrv}
\bibliography{paper}		
\balancecolumns
\end{document}